\documentclass[letterpaper,11pt,runningheads]{llncs}
\usepackage[T1,T5]{fontenc}
\usepackage{amssymb,theorem}
\usepackage{amsmath,array,amscd,pstricks}
\usepackage[mathscr]{eucal}
\usepackage{fullpage}
\DeclareMathOperator{\End}{End}
\DeclareMathOperator{\Frob}{Frob}

\newcommand{\Z}{\ensuremath{\mathbb{Z}}}

\newcommand{\K}{\ensuremath{\mathcal{K}}}
\newcommand{\brho}{\mathfrak{f}}
\newcommand{\f}{\ensuremath{\mathbb{F}}}

\newcommand{\ot}{\leftarrow}
\newcommand{\ind}{\hspace*{0.5cm}}
 {\theoremstyle{break}
 \theorembodyfont{\ttfamily}
 \newtheorem{algo}{Algorithm}}

\newif\ifblind

\title{Four-Dimensional Gallant-Lambert-Vanstone Scalar Multiplication}
\titlerunning{Four-Dimensional GLV Method}
\author{Peter Birkner\inst{1} \and Patrick Longa\inst{2} \and Francesco Sica\inst{3}
}

\institute{%
Bundesamt f\"ur Sicherheit in der Informationstechnik (BSI),\\
Godesberger Allee 185-189,
53175 Bonn,
Germany\\
\email{peter.birkner@bsi.bund.de}
\and
Microsoft Research,\\
One Microsoft Way, Redmond, WA 98052, USA\\
\email{plonga@microsoft.com}
\and
    Via Toscana 50,
    Prata (GR),
    58024 Italy\\
    \email{fracrypto@gmail.com}
 }

\authorrunning{P. Birkner \and P. Longa \and F. Sica}

\ifblind
\author{Submission to Eurocrypt 2012}
\institute{}
\authorrunning{}
\fi

\begin{document}
\maketitle \thispagestyle{plain}
\begin{abstract}

The GLV method of Gallant, Lambert and
Vanstone~(CRYPTO 2001) computes any multiple $kP$ of a point $P$ of prime order $n$ lying on an elliptic curve with a low-degree endomorphism $\Phi$ (called GLV curve) over $\mathbb{F}_p$ as
\[
kP = k_1P + k_2\Phi(P), \quad\text{with } \max\{|k_1|,|k_2|\}\leq C_1\sqrt n
\]
for some explicit constant $C_1>0$. Recently, Galbraith, Lin and Scott (EUROCRYPT 2009) extended this method to all curves over $\mathbb{F}_{p^2}$ which are twists of curves defined over $\mathbb{F}_p$. 
We show in this work how to merge the two approaches in order to get, for twists of any GLV curve over $\mathbb{F}_{p^2}$, a four-dimensional decomposition together with fast endomorphisms $\Phi, \Psi$ over $\mathbb{F}_{p^2}$ acting on the group generated by a point $P$ of prime order $n$, resulting in a proved decomposition for any scalar $k\in[1,n]$
$$
kP=k_1P+ k_2\Phi(P)+ k_3\Psi(P) + k_4\Psi\Phi(P)\quad \text{with } \max_i (|k_i|)< C_2\, n^{1/4}
$$
for some explicit $C_2>0$. Furthermore, taking the best $C_1, C_2$, we get $C_2/C_1<408$, independently of the curve, ensuring a constant relative speedup.

We also derive new families of GLV curves, corresponding to those curves with degree 3 endomorphisms.

\vspace{\baselineskip}
\end{abstract}

\section{Introduction}

The Gallant-Lambert-Vanstone (GLV) method is a generic method to speed up computation on some elliptic curves over fields of large characteristic. Given a curve with a point $P$ of prime order $n$, it consists essentially in an algorithm to find a decomposition of an arbitrary scalar multiplication $kP$ for $k\in[1,n]$ into two scalar multiplications with the new scalars having only about half the original bits. We call such a method two-dimensional, since if scalar multiplications can be parallelized, then a twofold performance speedup can be achieved.

Whereas the original GLV method as defined in~\cite{GLV01} works on curves over $\mathbb{F}_p$ with an endomorphism of small degree (GLV curves), Galbraith-Lin-Scott (GLS) in~\cite{GLS09} have shown that over $\mathbb{F}_{p^2}$ one can expect to find many more such curves by basically exploiting the action of the Frobenius endomorphism. One can therefore expect that on the particular GLV curves, this new insight will lead to improvements over $\mathbb{F}_{p^2}$. Indeed the GLS article itself considers fourfold speedups on GLV curves with nontrivial automorphisms (corresponding to the degree one cases) but leaves the other cases open to investigation.

Recently a paper by Zhou, Hu, Xu and Song~\cite{ZHXS10} has shown that it is possible to combine the two approaches by introducing a three-dimensional version of the GLV method (thus getting three scalars with a threefold speedup), which seems to be working to a certain degree, with however no justification but through practical implementations.

In contrast, we would like to show that the most natural understanding of their ideas is in four dimensions, where we are then able to construct, for the same curves and fast endomorphisms $\Phi, \Psi$ over $\mathbb{F}_{p^2}$ acting on a cyclic group generated by a point $P$ of prime order $n$, a proved decomposition for any scalar $k\in[1,n]$
$$
kP=k_1P+ k_2\Phi(P)+ k_3\Psi(P) + k_4\Psi\Phi(P)\quad \text{with } \max_i (|k_i|)< C n^{1/4}
$$
for some explicitly computable $C$. If parallel computation is available, then the computation of $kP$ can possibly be implemented up to four times as fast as a traditional scalar multiplication. It recently came to our attention that Hu, Longa and Xu~\cite{HLX11} also provided a similar bound in the case of curves with $j$-invariant 0. Our analysis supplements theirs by considering all GLV curves, where we provide an unified treatment.

The way to prove this bound is to study the kernel lattice of the GLV reduction map in dimension four. The LLL algorithm~\cite{LLL82} then will find a suitable reduced basis together with a useful bound to deduce $C$. In the last part of the article, we develop another approach which gives a reduced basis faster than the LLL algorithm together with a much better value for $C$. 
Indeed our reduction algorithm runs in $O(\log^2n)$ compared to $O(\log^3n)$ for LLL and the improved $C=O(\sqrt s)$ compared to the value obtained with LLL which is only $\Omega(s^{3/2})$. This allows us to prove that the relative speedup in going from a two-dimensional to a four-dimensional GLV method is independent of the curve.

\section{The GLV Method}
\label{S:2}

In this section we briefly summarize the GLV method following~\cite{SCQ02}. Let $E$ be an
elliptic curve defined over a finite field $\mathbb{F}_q$ and $P$ be a
point of this curve with prime order $n$ such that the cofactor
$h=\#E(\mathbb{F}_q)/n$ is small, say $h\leq4$. Let us consider $\Phi$
a non trivial endomorphism defined over $\mathbb{F}_q$ and $X^2+rX+s$
its characteristic polynomial. In all the examples $r$ and $s$ are
actually small fixed integers and $q$ is varying in some family. By
hypothesis there is only one subgroup of order $n$ in $E(\mathbb{F}_q)$, implying that $\Phi(P)= \lambda P$ for some
$\lambda\in [0,n-1]$, since $\Phi(P)$ has order dividing the prime $n$. 
In particular, $\lambda$ is obtained as a root of $X^2+rX+s$
modulo $n$.

Define the group homomorphism (the GLV reduction map)
\[ \begin{array}{rcl}
 \brho\colon\Z\times \Z &\to &\Z/n \\
 (i,j) & \mapsto & i+ \lambda j \pmod n\enspace.
 \end{array}
 \]

Let $\K=\ker\brho$. It is a sublattice of $\Z\times\Z$ of rank 2 since the quotient is finite. Let $\Bbbk>0$ be a constant (depending on the curve) such that we can find 
$v_1, v_2$ two linearly independent vectors of \K\ satisfying
$\max\{\left| v_1\right|, \left|
v_2\right|\}<~\Bbbk \sqrt n$, where
$\left|\vphantom{v_1}\cdot\right|$ denotes the 
rectangle norm\footnote{The rectangle norm of $(x,y)$
is by definition $\max(|x|,|y|)$. As remarked in~\cite{SCQ02}, we can replace it by any other metric norm. We will use the term ``short" to denote smallness in the rectangle norm.}. Express
 \[(k,0)= \beta_1v_1 + \beta_2v_2\enspace, \]
where $\beta_i\in\mathbb{Q}$. Then round $\beta_i$ to the nearest
integer $b_i=\lfloor\beta_i\rceil= \lfloor\beta_i+1/2\rfloor$ and let
$v= b_1v_1+ b_2v_2$. Note that $v\in\K$ and that
$u\overset{\text{def}}{=} (k,0)-v$ is short. Indeed by the triangle
inequality we have that
\[\left| \vphantom{v_1}u\right| \leq \frac{\left| v_1\right|+ \left| v_2 \right|}2
<\Bbbk \sqrt n\enspace.
\]
If we set $(k_1,k_2)=u$, then we get $k \equiv k_1+k_2\lambda \pmod n$ or
equivalently $kP = k_1P + k_2\Phi(P)$, with $\max
(|k_1|,|k_2|)<\Bbbk \sqrt n$. 

In~\cite{SCQ02}, the optimal value of $\Bbbk$ (with respect to large values of $n$, i.e. large fields, keeping $X^2+rX+s$ constant) is determined. Let $\Delta=r^2-4s$ be the discriminant of the characteristic polynomial of $\Phi$. Then the optimal $\Bbbk$ is given by the following result\footnote{There is a mistake in~\cite{SCQ02} in the derivation of $\Bbbk$ for odd values of $r$. This affects~\cite[Corollary~1]{SCQ02} for curves $E_2$ and $E_3$, where the correct values of $\Bbbk$ are respectively $2/3$ and $4\sqrt 2/7$ .} .

\begin{theorem}[\protect{\cite[Theorem~4]{SCQ02}}] 
\label{T:min}
Assuming $n$ is the norm of { an element} of $\Z[\Phi]$, then the optimal value of $\Bbbk$ is

$$
\Bbbk = \begin{cases}
 \dfrac{\sqrt s}2 \Bigl(1+\dfrac1{|\Delta|}\Bigr), &\text{if $r$ is odd,}\\
\rule{0pt}{2em}
\dfrac{\sqrt s}2 \sqrt{1+\dfrac4{|\Delta|}}, &\text{if $r$ is even.}
\end{cases}
$$

\end{theorem}

\section{The GLS Improvement}
\label{S:GLS}

In 2009, Galbraith, Lin and Scott~\cite{GLS09} realised that we don't need to have $\Phi^2+r\Phi+s=0$ in $\End(E)$ but only in a subgroup of $E(\mathbb{F})$ for a specific finite field $\mathbb{F}$. In particular, considering $\Psi=\Frob_p$ the $p$-Frobenius endomorphism of a curve $E$ defined over $\mathbb{F}_p$, we know that $\Psi^m(P)=P$ for all $P\in E(\mathbb{F}_{p^m})$. While this says nothing useful if $m=1,2$, it does offer new nontrivial relations for higher degree extensions. The case $m=4$ is particularly useful here.

In this case if $P\in E(\mathbb{F}_{p^4}) \backslash E(\mathbb{F}_{p^2})$, then 
$
\Psi^2(P) = -P
$ and hence on the subgroup generated by $P$, $\Psi$ satisfies the equation $X^2+1=0$. This implies that if $\Psi(P)$ is a multiple of $P$ (which happens as soon as the order $n$ of $P$ is sufficiently large, say at least $2p$), we can apply the previous GLV construction and split again a scalar multiplication as $kP = k_1P + k_2\Psi(P)$, with $\max (|k_1|,|k_2|) = O(\sqrt n)$.  
Contrast this with the characteristic polynomial of $\Psi$ which is $X^2-a_pX+p$ for some integer $a_p$, a non-constant polynomial to which we cannot apply as efficiently the GLV paradigm.

For efficiency reasons however one does not work with $E/\mathbb{F}_{p^4}$ directly but with $E'/\mathbb{F}_{p^2}$ isomorphic to $E$ \emph{over} $\mathbb{F}_{p^4}$ but not over $\mathbb{F}_{p^2}$, that is, a quadratic twist over $\mathbb{F}_{p^2}$. In this case, it's possible that $\#E'(\mathbb{F}_{p^2})=n\geq (p-1)^2$ be prime. Furthermore, if $\psi \colon E'\to E$ is an isomorphism defined over $\mathbb{F}_{p^4}$, then the endomorphism $\Psi = \psi \Frob_p \psi^{-1} \in \End(E')$ satisfies the equation $X^2+1=0$ and if $p\equiv 5 \pmod 8$ it can be defined over $\mathbb{F}_p$.

This idea is at the heart of the GLS approach, but it only works for curves over $\mathbb{F}_{p^m}$ with $m>1$, therefore it does not generalise the original GLV method but rather complements it.

\section{Examples}

We give a few examples of GLV curves, which are curves defined over $\mathbb{C}$ with complex multiplication by an quadratic integer of small norm, corresponding to an endomorphism $\phi$ of small degree\footnote{By small we mean really small, usually less than 5. In particular, for cryptographic applications, the degree is much smaller than the field size.}. 
  They make up an exhaustive list, up to isomorphism, in increasing order of endomorphism degree up to degree 3. While the first four examples appear in the previous literature, the next ones (degree 3) are new and have been computed with the Stark algorithm~\cite{S73}.

\begin{example}
Let $p\equiv 1\pmod 4$ be a prime. Define an elliptic curve $E_1$ over $\f_p$  by
\[ y^2=x^3+ax\enspace. \]
If $\beta$ is an element of order 4, then the map $\phi$ defined in
the affine plane by
\[\phi(x,y)=(-x, \beta y)\enspace,\]
is an endomorphism of $E_1$ defined over $\f_p$ with
$\End(E_1)=\Z[\phi] \cong \Z[\sqrt{-1}]$, since $\phi$ satisfies the equation
\[\phi^2+1=0\enspace.\]
\end{example}
\begin{example}
Let $p\equiv 1\pmod 3$ be a prime. Define an elliptic curve $E_2$ over $\f_p$  by
\[ y^2=x^3+b\enspace. \]
If $\gamma$ is an element of order 3, then we have an endomorphism $\phi$ defined over $\f_p$  by
\[\phi(x,y)=(\gamma x, y)\enspace,\]
and 
$\End(E_2)=\Z[\phi] \cong\Z[\frac{1+\sqrt{-3}}{2}]$, since $\phi$ satisfies the
equation
\[\phi^2+\phi+1=0\enspace.\]
\end{example}
\begin{example}
Let $p> 3$ be a prime such that -7 is a quadratic residue modulo $p$.
Define an elliptic curve $E_3$ over $\f_p$  by
\[ y^2=x^3- \frac34 x^2 -2x -1\enspace. \]
If $\xi=(1+\sqrt{-7})/2$ and $a=(\xi-3)/4$, then we get the $\f_p$-endomorphism $\phi$
defined by
\[\phi(x,y)=\left(\frac{x^2-\xi}{\xi^2(x-a)}, \frac{y(x^2-2ax+\xi)}{\xi^3(x-a)^2}\right)\enspace,\]
and
$\End(E_3)=\Z[\phi]\cong \Z[\frac{1+\sqrt{-7}}{2}]$, since $\phi$ satisfies the
equation
\[\phi^2-\phi+2=0\enspace.\]
\end{example}
\begin{example}
Let $p> 3$ be a prime such that -2 is a quadratic residue modulo $p$.
Define an elliptic curve $E_4$  over $\f_p$  by
\[ y^2=4x^3- 30x -28 \]
together with the $\f_p$-endomorphism $\phi$ defined\footnote{We take the opportunity to correct a typo found and transmitted in many sources, where a $y$ factor was absent in the second coordinate. Its sign is irrelevant.} by
\[\phi(x,y)=\left(-\frac{2x^2+4x+9}{4(x+2)}, y\frac{2x^2+8x-1}{4\sqrt{-2}(x+2)^2}\right)\enspace.\]
We have
$\End(E_4)=\Z[\phi]\cong\Z[\sqrt{-2}]$ since $\phi$ satisfies the equation
\[\phi^2+2=0\enspace.\]
\end{example}

\begin{example}
Let $p>3$ be a prime such that $-11$ is a quadratic residue $\mod p$. We define the elliptic curve $E_5$ over $\f_p$
$$
y^2 =  x^3 - \frac{13824}{539} x + \frac{27648}{539}
$$
with $a=(1+\sqrt{-11})/2$ and the endomorphism $\phi$ defined by 
\begin{multline*}
\phi(x,y)=\\
\left( \frac{\left(-\frac{539}{5184} a + \frac{539}{1728}\right) x^{3} + \left(\frac{28}{27} a - \frac{35}{18}\right) x^{2} + \left(-\frac{92}{9} a + \frac{8}{3}\right) x + \frac{1728}{77} a + \frac{192}{77}}{\left(\frac{2695}{5184} a - \frac{539}{864}\right) x^{2} + \left(-\frac{217}{54} a + \frac{49}{18}\right) x + \frac{64}{9} a - \frac{4}{3}} , y\right. \\
\left. 
\frac{\left(\frac{3773}{373248} a - \frac{18865}{995328}\right) x^{3} + \left(-\frac{2695}{20736} a + \frac{539}{3456}\right) x^{2} + \left(\frac{7}{432} a - \frac{91}{144}\right) x + \frac{20}{27} a + \frac{1}{9}}{\left(-\frac{18865}{1492992} a + \frac{116963}{995328}\right) x^{3} + \left(\frac{7007}{20736} a - \frac{539}{432}\right) x^{2} + \left(-\frac{791}{432} a + \frac{581}{144}\right) x + \frac{74}{27} a - \frac{35}{9}} 
\right)
\end{multline*}
such that $\End(E_5)=\mathbb{Z}[\phi] \cong \mathbb{Z}[\frac{1+\sqrt{-11}}2]$. The characteristic polynomial of $\phi$ is
$$
\phi^2 - \phi +3 =0 \enspace.
$$
\end{example}

\begin{example}
Let $p>3$ be a prime such that $-3$ is a quadratic residue $\mod p$. We define the elliptic curve $E_6$
over $\f_p$
$$
y^2 = x^3 - \frac{3375}{121}x + \frac{6750}{121}
$$
with the endomorphism $\phi$ defined by 
\begin{multline*}
\phi(x,y)=
\left( -\frac{1331x^3-10890x^2+81675x-189000}{33(11x-45)^2}, \right. \\
\left. y\, \frac{1331x^3 - 16335x^{2} + 7425x + 43875}{3\sqrt{-3} (11x-45)^3} \right)
\end{multline*}
such that\footnote{This is the first example where the endomorphism ring is not the maximal order of its field of fractions. It can be summarily seen as follows:  $\End(E)\supseteq \mathbb{Z}[\sqrt{-3}]$. If not equal, then it must be the full ring of integers $\mathbb{Z}[\frac{1+\sqrt{-3}}2]$. This would imply that $j=0$, as there is only $h(-3)=1$ isomorphism class of elliptic curves with complex multiplication by $\mathbb{Z}[\frac{1+\sqrt{-3}}2]$, given in Example~2 (see~\cite{S73} for an abridged description of the theory of complex multiplication). This is clearly not the case here. Alternatively, one can see that there would exist a nontrivial automorphism (a primitive cube root of unity) corresponding to $\frac{-1+\sqrt{-3}}2$. A direct computation then shows this is impossible.} $\End(E_6)=\mathbb{Z}[\phi] \cong \mathbb{Z}[\sqrt{-3}]$. The characteristic polynomial of $\phi$ is
$$
\phi^2 +3 =0 \enspace.
$$
\end{example}

\section{Combining GLV and GLS}
\label{S:4GLV}
Let $E/\mathbb{F}_p$ be a GLV curve. As in Section~\ref{S:GLS}, we will denote by $E'/\mathbb{F}_{p^2}$ a quadratic twist $\mathbb{F}_{p^4}$-isomorphic to $E$ via the isomorphism $\psi \colon E'\to E$. We also suppose that $\# E'(\mathbb{F}_{p^2}) = nh$ where $n$ is prime and $h\leq 4$. We then have the two endomorphisms of $E'$, $\Psi = \psi \Frob_p \psi^{-1}$ and $\Phi = \psi \phi \psi^{-1}$, with $\phi$ the GLV endomorphism coming with the definition of a GLV curve. They are both defined over $\f_{p^2}$, since if $\sigma$ is the nontrivial Galois automorphism of $\f_{p^4}/\f_{p^2}$, then $\psi^\sigma = -\psi$, so that $\Psi^\sigma= \psi^\sigma \Frob_p^\sigma \bigl(\psi^{-1}\bigr)^\sigma = (-\psi)\Frob_p(-\psi^{-1}) = \Psi$, meaning that $\Psi\in\End_{\f_{p^2}}(E')$. Similarly for $\Phi$, where we are using the fact that $\phi\in\End_{\f_p}(E)$. Notice that  $\Psi^2+1=0$ and that $\Phi$ has the same characteristic polynomial as $\phi$. Furthermore, since we have a large subgroup $\langle P \rangle \subset E'(\mathbb{F}_{p^2})$ of prime order, $\Phi(P)=\lambda P$ and $\Psi(P)=\mu P$ for some $\lambda,\mu\in [1, n-1]$. We will assume that $\Phi$ and $\Psi$, when viewed as algebraic integers, generate disjoint quadratic extensions of $\mathbb{Q}$. In particular, we are not dealing with Example~1, but this can be treated separately with a quartic twist, as was hinted in~\cite{GLS09}. 

Consider the biquadratic (Galois of degree 4, with Galois group $\mathbb{Z}/2\times\mathbb{Z}/2$) number field $K=\mathbb{Q}(\Phi,\Psi)$. Let $\mathfrak{o}_K$ be its ring of integers. The following analysis is inspired by Sica, Ciet and Quisquater~\cite[Section~8]{SCQ02}. 

We have $\mathbb{Z}[\Phi, \Psi] \subseteq \mathfrak{o}_K$. Since the degrees of $\Phi$ and $\Psi$ are much smaller when compared to $n$, the prime $n$ is unramified in $K$ and the existence of $\lambda$ and $\mu$ above means that $n$ splits in $\mathbb{Q}(\Phi)$ and $\mathbb{Q}(\Psi)$, namely that $n$ splits completely in $K$. There exists therefore a prime ideal $\mathfrak n$ of $\mathfrak{o}_K$ dividing $n\mathfrak{o}_K$, such that its norm is $n$. We can also suppose that $\Phi\equiv \lambda \pmod{\mathfrak n}$ and $\Psi\equiv \mu \pmod{\mathfrak n}$. The four-dimensional GLV (4-GLV) method works as follows.

Consider the 4-GLV reduction map $F$ defined by
\[ \begin{array}{rcl}
 F\colon\Z^4 &\to &\Z/n \\
 (x_1, x_2, x_3,x_4) & \mapsto & x_1+x_2\lambda+x_3\mu+x_4\lambda\mu \pmod n\enspace.
 \end{array}
 \]

If we can find four linearly independent vectors $v_1,\dots,v_4 \in \ker F$, with $\max_i |v_i| \leq C n^{1/4}$ for some constant $C>0$, then
for any $k\in[1,n-1]$ we write
\[ (k,0,0,0) = \sum_{j=1}^4 \beta_j v_j \enspace,\]
with $\beta_j\in\mathbb Q$. As in the GLV method one sets $v= \sum_{j=1}^4 \lfloor \beta_j \rceil v_j$ and
\[u=
(k,0,0,0)-v= (k_1, k_2, k_3, k_4) \enspace. \]

We then get
\begin{equation}
\label{E:1}
kP=k_1P+ k_2\Phi(P)+ k_3\Psi(P) + k_4\Psi\Phi(P)\quad \text{with } \max_i (|k_i|)\leq 2C n^{1/4} \enspace.
\end{equation}

We focus next on the study of $\ker F$ in order to find a reduced basis $v_1, v_2, v_3, v_4$ with an explicit $C$. We can factor the 4-GLV map $F$ as

\begin{align*}
&\Z^4   \xrightarrow{\phantom{\mapsto} f \phantom{\mapsto}} &
&\Z[\Phi,\Psi] \xrightarrow[\text{\raisebox{.5ex}[0cm]{$\mod
\mathfrak{n}\cap\Z[\Phi,\Psi]$}}]{\text{reduction}}
\hspace{.5cm}\Z/n  \\
(x_1,x_2, &x_3,x_4) \longmapsto & &x_1+x_2\Phi+ x_3\Psi+  x_4\Phi\Psi  
\longmapsto   x_1+x_2\lambda+x_3\mu+x_4\lambda\mu \\
& && \phantom{x_1+x_2\Phi+ x_3\Psi+  x_4\Phi\Psi  
\longmapsto   x_1+x_2\lambda+} \pmod n
\enspace.
\end{align*}

Notice that the kernel of the second map (reduction mod $\mathfrak{n}\cap\Z[\Phi,\Psi]$) is exactly $\mathfrak{n}\cap\Z[\Phi,\Psi]$. This can be seen as follows. The reduction map factors as 
$$
\mathbb{Z}[\Phi,\Psi] \longrightarrow \mathfrak{o}_K \longrightarrow \mathfrak{o}_K / \mathfrak{n} \cong \mathbb{Z}/n
$$
where the first arrow is inclusion, the second is reduction mod $\mathfrak{n}$, corresponding to reducing the $x_i$'s mod $\mathfrak{n}\cap\mathbb{Z}= n\mathbb{Z}$ and using $\Phi\equiv\lambda, \Psi \equiv\mu \pmod{\mathfrak{n}}$. But the kernel of this map consists precisely of elements of $\mathbb{Z}[\Phi,\Psi]$ which are in $\mathfrak{n}$, and that is what we want. 

Moreover, since the reduction map is surjective, we obtain an isomorphism $ \Z[\Phi,\Psi]/\mathfrak{n}\cap\Z[\Phi,\Psi] \cong \mathbb{Z}/n$ which says that the index of $\mathfrak{n}\cap\Z[\Phi,\Psi]$ inside
$\Z[\Phi,\Psi]$ is $n$. Since the first map $f$ is an isomorphism, we get that 
$\ker F = f^{-1} (\mathfrak{n}\cap\Z[\Phi,\Psi])$ and that $\ker F$ has index $[\mathbb{Z}^4 \colon \ker F]=n$ inside $\mathbb{Z}^4$.

We can also produce a basis of $\ker F$ by the following observation. Let $\Phi' = \Phi-\lambda$, $\Psi' = \Psi - \mu$, hence $\Phi'\Psi'= \Phi\Psi - \lambda \Psi -\mu\Phi + \lambda\mu$. In matrix form,

$$
\begin{pmatrix}
1\\
\Phi'\\
\Psi'\\
\Phi'\Psi'
\end{pmatrix} =
\begin{pmatrix}
1 & 0 & 0 & 0 \\
-\lambda &1 & 0 & 0\\
-\mu & 0 & 1 & 0 \\
\lambda\mu & -\mu & -\lambda & 1
\end{pmatrix}
\begin{pmatrix}
1\\
\Phi\\
\Psi\\
\Phi\Psi
\end{pmatrix}
$$

Since the determinant of the square matrix is $1$, we deduce that $\mathbb{Z}[\Phi,\Psi] = \mathbb{Z}[\Phi', \Psi']$. But in this new basis, we claim that 
$$
\mathfrak{n}\cap\Z[\Phi',\Psi'] = n\mathbb{Z} + \mathbb{Z}\Phi' + \mathbb{Z}\Psi' + \mathbb{Z}\Phi'\Psi' \enspace.
$$

Indeed, reverse inclusion ($\supseteq$) is easy since $\Phi',\Psi', \Phi'\Psi' \in \mathfrak{n}$ and so is $n$, because $\mathfrak{n}$ divides $n\mathfrak{o}_K$ is equivalent to  $\mathfrak{n} \supseteq n\mathfrak{o}_K$. On the other hand, the index of both sides in $\mathbb{Z}[\Phi',\Psi']$ is $n$, which can only happen, once an inclusion is proved, if the two sides are equal. Using the isomorphism $f$, we see that a basis of $\ker F\subset \mathbb{Z}^4$ is therefore given by 
$$
w_1= (n,0,0,0), w_2= (-\lambda, 1 ,0,0), 
w_3 = (-\mu, 0, 1, 0), 
w_4 = ( \lambda\mu, -\mu, -\lambda, 1) \enspace.
$$

The LLL algorithm~\cite{LLL82} then finds, for a given basis $w_1,\dots, w_4$ of
$\ker F$, a reduced\footnote{The estimates are usually given for the Euclidean norm of the vectors. But it is easy to see that the rectangle norm is upper bounded by the Euclidean norm.} basis $v_1,\dots,v_4$ in polynomial time (in the logarithm of the norm of
the $w_i$'s) such that (cf.~\cite[Theorem 2.6.2 p.85]{Co})
\begin{equation}\label{E2}
 \prod_{i=1}^4 |v_i|
 \leq 8\, [\mathbb{Z}^4 \colon \ker F] = 8n\enspace.
\end{equation}

\begin{lemma}
\label{L:genlownumfield}
Let
\[ \begin{array}{rcl}
 \mathscr{N}\colon\Z^4 &\to &\Z \\
 (x_1,x_2,x_3,x_4) & \mapsto & \displaystyle\sum_{\substack{i_1,i_2,i_3,i_4\geq 0\\ i_1+i_2+i_3+i_4=4}}
 b_{i_1,i_2,i_3,i_4} x_1^{i_1}x_2^{i_2}x_3^{i_3} x_4^{i_4}
 \end{array}
\]
be the norm of an element $x_1+x_2\Phi+x_3\Psi+x_4\Phi\Psi \in \mathbb{Z}[\Phi,\Psi]$,
where the $b_{i_1,i_2,i_3,i_4}$'s lie in $\Z$. Then, for any nonzero
$v \in\ker F$, one has
 \begin{equation}\label{E3}
  |v| \geq
\frac{n^{1/4}}{\Bigl(\displaystyle\sum_{\substack{i_1,i_2,i_3,i_4\\
i_1+i_2+i_3+i_4=4}} |b_{i_1,i_2,i_3,i_4}|\Bigr)^{1/4}} \enspace.
\end{equation}
\end{lemma}
\begin{proof}
For
$v\in\ker F$ we have
$\mathscr{N}(v)\equiv 0 \pmod n$ and if $v\neq 0$ we must
therefore have $|\mathscr{N}(v)|\geq n$. On the other hand, if we
did not have~\eqref{E3}, then every component of $v$ would be
strictly less than the right-hand side and plugging this upper bound
in the definition of $|\mathscr{N}(v)|$ would yield a quantity $<n$, a
contradiction. \hfil\qed
\end{proof}

Let $B$ be the denominator of the right-hand side of~\eqref{E3},
then~\eqref{E2} and~\eqref{E3} imply that
\begin{equation}
\label{E:reduced}
|v_i| \leq 8B^{3}\, n^{1/4} \quad i=1,2,3,4 \enspace.
\end{equation}


\begin{remark}
In our case, where $\Psi^2+1=0$ and $\Phi^2+r\Phi +s=0$, we get as norm function
\begin{multline*}
x_{1}^{4} + s^2 x_2^4 +  x_{3}^{4} + s^2 x_4^4 
-2r x_1^3x_2 - 2rs x_1 x_2^3  -2r x_3^3x_4  -2rs x_3 x_4^3 + \\
(r^2+2s) x_1^2 x_2^2 + 2x_1^2x_3^2+  (r^2-2s) x_1^2x_4^2 +  (r^{2}-2s) x_{2}^{2} x_{3}^{2} + 2s^2x_2^2x_4^2 + (r^{2}+2s) x_{3}^{2} x_{4}^{2} \\
-2rx_1^2x_3x_4 - 2rs x_2^2x_3x_4 - 2r x_1x_2x_3^2 - 2rs x_1x_2x_4^2
+ 8s x_1x_2x_3x_4 \enspace,
\end{multline*}
and therefore
\begin{equation}
\label{E:B}
B=
\bigl(4+4s^2 + 8s + 8|r| + 8 |r| s + 2 (r^2+2s) + 2 |r^2-2s|\bigr)^{1/4} \enspace.
\end{equation}

\end{remark}

From~\eqref{E:1} and \eqref{E:reduced} we have proved the following theorem.
\begin{theorem}
\label{T:1}
Let $E/\mathbb{F}_p$ be a GLV curve and $E'/\mathbb{F}_{p^2}$ a twist, together with the two efficient endomorphisms $\Phi$ and $\Psi$, where everything is defined as at the start of Section~\ref{S:4GLV}. Suppose that the minimal polynomial of $\Phi$ is $X^2+rX+s=0$. Let $P\in E'(\mathbb{F}_{p^2})$ a generator of the large subgroup of prime order $n$. There exists an efficient algorithm, which for any $k\in [1,n]$ finds integers $k_1, k_2, k_3, k_4$ such that
$$
kP=k_1P+ k_2\Phi(P)+ k_3\Psi(P) + k_4\Psi\Phi(P)\quad \text{with } \max_i (|k_i|)\leq 16 B^3 n^{1/4}
$$
and
$$
B=
\bigl(4+4s^2 + 8s + 8|r| + 8 |r| s + 2 (r^2+2s) + 2 |r^2-2s|\bigr)^{1/4} \enspace.
$$
\end{theorem}

\label{S:5}

\section{A Tale of Two Cornacchia Algorithms}

In view of the fact that the LLL algorithm is rather inefficient compared to other dedicated algorithms in dimension less than five (running in $O(\log^3 n)$), we can ask ourselves if we can sharpen the bound of Theorem~\ref{T:1} and provide an explicit description of a simpler algorithm to find a short basis of $\ker F$. 
This is the scope of the the present section. Our algorithm has a running time of $O(\log^2 n)$, and will produce a proved bound greatly improving the $16B^3$ of Theorem~\ref{T:1}.

The idea is to modify the original GLV approach which finds a short basis using an extended Euclidean algorithm. We find that in this case we need to perform two such algorithms, one in $\mathbb{Z}$, like in the GLV original paper, the other one in $\mathbb{Z}[i]$, the Gaussian integers. 
The main difficulty here lies in the correct choice of the remainders in the Gaussian gcd algorithm, since we don't have a canonical way to choose a ``positive" one.

In contrast to Section~\ref{S:5}, where we worked with generic endomorphisms $\Phi,\Psi$ generating a biquadratic field, we will strongly use here the fact that $\Psi^2+1=0$. We will denote indifferently by the letter $i$ the usual imaginary root of unity in $\mathbb{C}$, the integer mod $n$ such that $\Psi (P) = i P$, as well as the endomorphism $\Psi$. In particular, we let, for $z=a+ib\in\mathbb{Z}[i]$, $zP=aP+ibP= aP + b \Psi(P)$. The context in which we are referring to one or the other of these interpretations will be clear each time. These differences notwithstanding, we suppose that we are set as in the first paragraph of Section~\ref{S:5}. 

\subsection{The Euclidean Algorithm in $\mathbb{Z}$}
The first step is to find $\nu=a+ib \in \mathbb{Z}[i]$ such that $|\nu|^2 = a^2+b^2 = n$, i.e. a Gaussian prime above $n$. Recall that $n$ splits in $\mathbb{Z}[i]$. Let $\nu = a+ib$ a prime above $n$. We can furthermore assume that $\nu P = aP + bi P = aP + b \Psi(P) =0$, since $\nu\bar\nu P = nP=0$ and hence either $\bar\nu P$ is a nonzero multiple of $P$ and therefore $\nu P=0$, or else we $\bar\nu P=0$, so that in any case one of the Gaussian primes (WLOG $\nu$) above $n$ will have $\nu P=0$. We can find $\nu$ by Cornacchia's algorithm~\cite[Section 1.5.2]{Co}, which is a truncated form of the GLV algorithm. For completeness and consistency with what will follow, we recall how this is done. 

Let $\mu\in[1,n]$ such that $\mu \equiv i \pmod n$, with $i$ being defined by $\Psi(P)=iP$. Actually, in the GLS approach~\cite{GLS09}, it has been pointed out that this value of $\mu$ can be readily computed from $\#E(\f_p)$. The extended Euclidean algorithm to compute the gcd of $n$ and $\mu$ produces three terminating sequences of integers $(r_j)_{j\geq 0}, (s_j)_{j\geq 0}$ and $(t_j)_{j\geq 0}$ such that 
\begin{equation}
\label{E:gcdmatrix}
\begin{pmatrix}
r_{j+2} & s_{j+2} & t_{j+2} \\
r_{j+1} & s_{j+1} & t_{j+1}
\end{pmatrix}
=
\begin{pmatrix}
-q_{j+1} & 1\\
1 & 0
\end{pmatrix}
\begin{pmatrix}
r_{j+1} & s_{j+1} & t_{j+1}\\
r_{j} & s_{j} & t_{j}
\end{pmatrix}\enspace,
\quad j\geq 0
\end{equation}
for some integer $q_{j+1}>0$ and initial data
\begin{equation}
\label{E:initial}
\begin{pmatrix}
r_{1} & s_{1} & t_{1}\\
r_{0} & s_{0} & t_{0}
\end{pmatrix}
=
\begin{pmatrix}
\mu & 0 & 1\\
n & 1 & 0
\end{pmatrix} \enspace.
\end{equation}
This means that at step $j\geq 0$,
$$
r_j = q_{j+1} r_{j+1} + r_{j+2}
$$
and similarly for the other sequences. The sequence $(q_j)_{j\geq 1}$ is \emph{uniquely} defined by imposing that the previous equation be the integer division of $r_j$ by $r_{j+1}$. In other terms,
$q_{j+1} = \lfloor r_j/r_{j+1} \rfloor$. This implies by induction that all the sequences are well defined in the integers, together with the following properties.
\begin{lemma}
\label{L:gcdprop}
The sequences $(r_j)_{j\geq 0}, (s_j)_{j\geq 0}$ and $(t_j)_{j\geq 0}$ defined by~\eqref{E:gcdmatrix} and~\eqref{E:initial} with $q_{j+1} = \lfloor r_j/r_{j+1} \rfloor$ satisfy the following properties, valid for all $j\geq 0$.
\begin{enumerate}
\item $r_j > r_{j+1} \geq 0$ and $q_{j+1} \geq 1$,
\item $(-1)^j s_j \geq 0$ and $|s_j| < |s_{j+1}|$ (this last inequality valid for $j\geq 1$),
\item $(-1)^{j+1} t_j \geq 0$ and $|t_j| < |t_{j+1}|$,
\item $s_{j+1}r_j - s_jr_{j+1} = (-1)^{j+1} r_1$,
\item $t_{j+1} r_j - t_j r_{j+1} = (-1)^j r_0$,
\item $r_0 s_j + r_1 t_j = r_j$.
\end{enumerate}
\end{lemma}

These properties lie at the heart of the original GLV algorithm. They imply in particular via 1. that the algorithm terminates (once $r_j$ reaches zero), and that it has $O(\log n)$ steps, as 
$r_j = q_{j+1}r_{j+1} + r_{j+2} \geq r_{j+1} + r_{j+2} > 2 r_{j+2}$. Note that 1., 2. \& 3. imply that 4. \& 5. can be rewritten in our case respectively as 
\begin{equation}
\label{E:pos}
|s_{j+1} r_j| + |s_j r_{j+1}| = \mu  \quad \text{and} \quad |t_{j+1} r_j| + |t_j r_{j+1}| = n \enspace.
\end{equation}
The Cornacchia (as well as the GLV) algorithm doesn't make use of the full sequences $(r_j), (s_j)$ and $(t_j)$ but rather stops at the $m\geq 0$ such that $r_{m}\geq \sqrt n$ and $r_{m+1}< \sqrt n$. An application of~\eqref{E:pos} with $j=m$ yields $|t_{m+1} r_{m}| < n$ or $|t_{m+1}| < \sqrt n$.  Since by 6. we have $r_{m+1} - \mu t_{m+1} = ns_{m+1} \equiv 0 \pmod n$ we deduce that $r_{m+1}^2+ t_{m+1}^2=(r_{m+1}-\mu t_{m+1})(r_{m+1}+\mu t_{m+1}) \equiv 0 \pmod n$. Moreover $t_{m+1} \neq 0$ by 3. so that
$
0< r_{m+1}^2+t_{m+1}^2 < n + n = 2n
$
which therefore implies that $r_{m+1}^2+ t_{m+1}^2 = n$ and finally that $\nu = r_{m+1} - i t_{m+1}$.

We present here the pseudo-code of this Euclidean algorithm in $\mathbb{Z}$.

\begin{algo}[Cornacchia's GCD in $\Z$]
\label{A:1}
\rule{\linewidth}{1pt}
{\tt Input:} $n\equiv 1 \pmod 4$ prime, $1<\mu<n$ such that $\mu^2\equiv -1 \pmod n$.\\
{\tt Output:} $\nu=\nu_{(R)} + i \nu_{(I)}$ Gaussian prime dividing $n$, such that $\nu P=0$. \\
\rule{\linewidth}{.5pt}
\begin{enumerate}
\item {\bf initialize:}\\
$r_0 \ot n$, $r_1 \ot \mu$, $r_2 \ot n$, \\
$t_0 \ot 0$, $t_1 \ot 1$, $t_2 \ot 0$,\\
$q \ot 0$.
\item  {\bf main loop:}\\
{\tt while} $r_2^2\geq n$ {\tt do}\\
\ind $q \ot \lfloor r_0/r_1 \rfloor$,\\
\ind $r_2 \ot r_0 - q r_1$, $r_0 \ot r_1$, $r_1 \ot r_2$,\\
\ind $t_2 \ot t_0 - q t_1$, $t_0 \ot t_1$, $t_1 \ot t_2$.
\item {\bf return:}\\
$\nu= r_1 - i t_1$, $\nu_{(R)}= r_1$, $\nu_{(I)}=-t_1$
\end{enumerate}
\rule{\linewidth}{1pt}
\end{algo}

\subsection{The Euclidean Algorithm in $\mathbb{Z}[i]$}
In the previous subsection we have given a meaning to $zP$, where $z\in\mathbb{Z}[i]$, and we have seen how to construct $\nu$, a Gaussian prime such that $\nu P =0$. By identifying\footnote{It is important to keep in mind that this association is only an isomorphism of abelian groups ($\Z$-modules). However, $\Z[i]^2$ is also endowed with a structure of $\Z[i]$-module.} $(x_1, x_2, x_3, x_4) \in\mathbb{Z}^4$ with $(z_1, z_2)= (x_1+ix_3, x_2 +ix_4)\in \mathbb{Z}[i]^2$, we can rewrite the 4-GLV reduction map $F$ of Section~\ref{S:5} as (using the same letter $F$ by abuse of notation)
\[ \begin{array}{rcl}
 F\colon\mathbb{Z}[i]^2 &\to &\Z[i]/\nu \cong \Z/n \\
 (z_1, z_2) & \mapsto & z_1 +\lambda z_2 \pmod \nu\enspace.
 \end{array}
 \]
 This $F$ should be confronted to the map $\brho$ of Section~\ref{S:2}. In mimicking the GLV original paper~\cite{GLV01} we are drawn to applying the extended Euclidean algorithm (defined exactly as before, with integer divisions occurring in $\Z[i]$, henceforth denoted EGEA in short for extended Gaussian Euclidean algorithm) to the pair $(r_0, r_1)=(\lambda, \nu)$ if $\lambda \geq\sqrt2\, |\nu|$ and $(r_0, r_1)=(\lambda+n, \nu)$ otherwise (the latter case being exceptionally rare). We should note that 4., 5. \& 6. of Lemma~\ref{L:gcdprop} still hold and 1. holds in modulus (in particular the algorithm terminates). However, in the analysis of this algorithm, especially in~\cite{SCQ02}, a crucial r\^ole is played by~\eqref{E:pos}, realising a bound on  $|s_{j+1}r_j|$ and  $|s_jr_{j+1}|$  out of a bound on 

\begin{equation}
\label{E:cruxgauss}
 s_{j+1}r_j - s_jr_{j+1} = (-1)^{j+1}\nu
\end{equation}
in the present case. This fact, as we saw, stems from the alternating sign of the sequence $(s_j)$, which results from taking a canonical form of integer division with positive quotients $q_{j+1}$ and nonnegative remainders $r_{j+2}$, a property which is not available here. Nevertheless, we can still use a similar reasoning using~\eqref{E:cruxgauss}, provided that the arguments of $s_{j+1}r_j$ and $s_j r_{j+1}$ are not too close, so as to avoid a high degree of cancellation.
 

The first observation is that in the case of Gaussian integers there can be 2, 3 or 4 possible choices for a remainder in the $j$-th step of the integer division $r_j = q_{j+1} r_{j+1} + r_{j+2}$. It turns out that choosing at each step $j\geq0$ of the EGEA a remainder $r_{j+2}$ with smallest modulus will yield the following decomposition theorem.

\begin{theorem}
\label{T:3}
In the notations of Theorem~\ref{T:1}, the lattice reduction consisting of Cornacchia's algorithm in $\Z$ with positive remainders (Algorithm~\ref{A:1}) and in $\Z[i]$ with smallest remainders (Algorithm~\ref{A:2} and~\ref{A:3}) runs in $O(\log^2 n)$ binary operations and will result in a decomposition of any $k\in [1,n]$ into integers $k_1, k_2, k_3, k_4$ such that
$$
kP=k_1P+ k_2\Phi(P)+ k_3\Psi(P) + k_4\Psi\Phi(P)
$$
with
$$
 \max_i (|k_i|)< 103 \left(\sqrt{1+|r|+s}\right) n^{1/4} \enspace.
$$
\end{theorem}

We give here the pseudo-code of Cornacchia's Algorithm in $\Z[i]$ in two forms, working with complex numbers and separating real and imaginary parts.

\begin{algo}[\protect{Cornacchia's algorithm in $\Z[i]$} - compact form]
\label{A:2}
\rule{\linewidth}{1pt}
{\tt Input:} $\nu$ Gaussian prime dividing $n$ rational prime, $1<\lambda<n$ such that $\lambda^2+r\lambda +s \equiv 0 \pmod n$.\\
{\tt Output:} Two $\Z[i]$-linearly independent vectors $v_1$ \& $v_2$ of $\ker F \subset \Z[i]^2$ of rectangle norms $< 51.5 (\sqrt{1+|r|+s})\, n^{1/4}$. \\
\rule{\linewidth}{.5pt}
\begin{enumerate}
\item {\bf initialize:}\\
If $\lambda^2\geq 2n$ then\\
\ind $r_0 \ot \lambda$, \\
else\\
\ind $r_0 \ot \lambda+n$,\\
$r_1 \ot \nu$, $r_2 \ot n$, \\
$s_0 \ot 1$, $s_1 \ot 0$, $s_2 \ot 0$,\\
$q \ot 0$.
\item  {\bf main loop:}\\
{\tt while} $|r_2|^4(1+|r|+s)^2\geq n$ {\tt do}\\
\ind $q \ot \text{closest Gaussian integer to } r_0/r_1$,\\
\ind $r_2 \ot r_0 - q r_1$, $r_0 \ot r_1$, $r_1 \ot r_2$,\\
\ind $s_2 \ot s_0 - q s_1$, $s_0 \ot s_1$, $s_1 \ot s_2$.
\item {\bf return:}\\
$v_1= (r_0, -s_0)$, $v_2= (r_1, -s_1)$
\end{enumerate}
\rule{\linewidth}{1pt}
\end{algo}

\begin{algo}[\protect{Cornacchia's algorithm in $\Z[i]$} - real \& imaginary parts]
\label{A:3}
\rule{\linewidth}{1pt}
{\tt Input:} $\nu$ Gaussian prime dividing $n$ rational prime, $1<\lambda<n$ such that $\lambda^2+r\lambda +s \equiv 0 \pmod n$.\\
{\tt Output:} Four $\Z$-linearly independent vectors $v_1$, $v_2$, $v_3$ and $v_4 \in \ker F \subset \Z^4$ of rectangle norms $< 51.5 (\sqrt{1+|r|+s})\, n^{1/4}$. \\
\rule{\linewidth}{.5pt}
\begin{enumerate}
\item {\bf initialize:}\\
If $\lambda^2\geq 2n$ then\\
\ind $r_{0,(R)} \ot \lambda$, \\
else\\
\ind $r_{0,(R)} \ot \lambda+n$,\\
$r_{0,(I)} \ot 0$,\\
$r_{1,(R)} \ot \nu_{(R)}$, $r_{1,(I)} \ot \nu_{(I)}$,\\
$r_{2,(R)} \ot n$, $r_{2,(I)} \ot 0$, \\
$s_{0,(R)} \ot 1$, $s_{0,(I)} \ot 0$,\\
$s_{1,(R)} \ot 0$, $s_{1,(I)} \ot 0$, \\
$s_{2,(R)} \ot 0$, $s_{2,(I)} \ot 0$,\\
$q_{(R)} \ot 0$, $q_{(I)} \ot 0$.
\item  {\bf main loop:}\\
{\tt while} $(r_{2,(R)}^4+ 2r_{2,(R)}^2r_{2,(I)}^2 + r_{2,(I)}^4)(1+|r|+s)^2\geq n$ {\tt do}\\
\ind $q_{(R)} \ot \left\lceil \dfrac{r_{0,(R)} r_{1,(R)} + r_{0,(I)}r_{1,(I)}}{r_{1,(R)}^2 + r_{1,(I)}^2}\right \rfloor$,\\
\\
\ind $q_{(I)} \ot \left\lceil \dfrac{r_{0,(I)} r_{1,(R)} - r_{0,(R)}r_{1,(I)}}{r_{1,(R)}^2 + r_{1,(I)}^2}\right \rfloor$, \\
\ind $r_{2,(R)} \ot r_{0,(R)}- (q_{(R)}r_{1,(R)} - q_{(I)} r_{1,(I)})$,\\
\ind $r_{2,(I)} \ot r_{0,(I)}- (q_{(R)}r_{1,(I)} + q_{(I)} r_{1,(R)})$,\\
\ind $r_{0,(R)} \ot r_{1,(R)}$, $r_{1,(R)} \ot r_{2,(R)}$,\\
\ind $r_{0,(I)} \ot r_{1,(I)}$, $r_{1,(I)} \ot r_{2,(I)}$,\\
\ind $s_{2,(R)} \ot s_{0,(R)}- (q_{(R)}s_{1,(R)} - q_{(I)} s_{1,(I)})$,\\
\ind $s_{2,(I)} \ot s_{0,(I)}- (q_{(R)}s_{1,(I)} + q_{(I)} s_{1,(R)})$,\\
\ind $s_{0,(R)} \ot s_{1,(R)}$, $s_{1,(R)} \ot s_{2,(R)}$,\\
\ind $s_{0,(I)} \ot s_{1,(I)}$, $s_{1,(I)} \ot s_{2,(I)}$.
\item {\bf return:}\\
$v_1= (r_{0,(R)}, -s_{0,(R)}, r_{0,(I)}, -s_{0,(I)})$,
$v_2= (r_{1,(R)}, -s_{1,(R)}, r_{1,(I)}, -s_{1,(I)})$,\\
$v_3= (-r_{0,(I)}, s_{0,(I)}, r_{0,(R)}, -s_{0,(R)})$,
$v_4= (-r_{1,(I)}, s_{1,(I)}, r_{1,(R)}, -s_{1,(R)})$.
\end{enumerate}
\rule{\linewidth}{1pt}
\end{algo}

\section{Proof of Theorem~\ref{T:3}}

This section is devoted to proving that Algorithms~\ref{A:2} and~\ref{A:3} produce a reduced basis of $\ker F$ of rectangle norm $< 51.5 (\sqrt{1+|r|+s})\, n^{1/4}$. The proof of the decomposition of $k$ follows from the deduction recalled in Section~\ref{S:4GLV}.

Let us note first, about the running time, that it is known that the extended Euclidean algorithm runs in $O(\log^2 n)$ bits. The same analysis will also show that its Gaussian version runs in $O(\log^2n)$ bits, since its number of steps is also logarithmic. In short, this works as follows: if $b_j=\lfloor \log_2 (|r_j|) \rfloor$ (i.e. the bitsize of $|r_j|$), then step $j$ of the EGEA necessitates to find $q_{j+1}$ and then $r_{j+2}$. One can show that integer division of two $h$-bit Gaussian integers with a $\ell$-bit quotient runs in $O(h(\ell+1))$ binary operations.  Finding $q_{j+1}$ has therefore a runtime $O(b_j (c_{j+1}+1))$, where $c_{j+1}= \lfloor \log_2 (|q_{j+1}|) \rfloor = b_j-b_{j+1} +O(1)$. Similarly, knowing $q_{j+1}$, computing $r_{j+2}$ can be done in $O(b_{j+1}c_{j+1}) + O(b_{j+1}) = O(b_{j+1} (b_{j}-b_{j+1})) +O(b_{j+1})$.
If $S= O(\log n)$ is the number of steps of the EGEA, the total runtime is less than a constant times
$$
\sum_{j=0}^S
b_{j} (b_{j}-b_{j+1})+ b_{j}  = O( b_0 ^2 + b_0 S) = O(\log^2 n) \enspace.
$$

In the following, whenever $z\in\mathbb{C}^*$, its argument value $\arg(z)$ will be always chosen in $(-\pi, \pi]$. By \emph{lattice square} we mean a square of side length one with vertices in $\Z[i]$. We single out eight \emph{exceptional lattice squares}, which are those lattice squares with a vertex of modulus 1 (that is $\pm1$ or $\pm i$) but not containing the origin as a vertex.
Our analysis of the EGEA rests on the following lemmas.

\begin{lemma}[A geometric property of squares]
\label{L:geo}
There exists an absolute real constant $\theta\approx 2.45861$ (with $2\arctan2<\theta$) such that,
for any point $P$ of a lattice square, different from the vertices, letting $V_1$ be the closest vertex to $P$, there exists another vertex $V_2 \neq V_1$ with $\theta\leq \widehat{V_1PV_2} \leq \pi$. (Note that $V_1P \leq 1/\sqrt2$.)
\end{lemma}

\begin{figure}[!ht]
\begin{center}
\psset{unit=2.5}
\radians
\SpecialCoor
\pspicture(-1,0)(3,2)
\pspolygon(0,0)(2,0)(2,2)(0,2)
\psarc(1,3.229095){1.58451}{-2.25378}{-0.8878}
\psarc(-1.229095,1){1.58451}{-0.68298}{0.68298}
\psarc(1,-1.229095){1.58451}{0.8878}{2.25378}
\psarc(3.229095,1){1.58451}{2.45861}{-2.45861}
\psarc(-0.2291,-0.2291){2.240837}{0.102415}{1.46838}
\psarc(2.2291,2.2291){2.240837}{3.244}{-1.67321}
\psarc(2.2291,-0.2291){2.240837}{1.67321}{-3.244}
\psarc(-0.2291,2.2291){2.240837}{-1.46838}{-0.102415}
\degrees
\psarc[linestyle=dotted](0,0){1.4142}{0}{90}
\psarc[linestyle=dotted](2,2){1.4142}{180}{270}
\psarc[linestyle=dashed](2,0){1.4142}{90}{180}
\psarc[linestyle=dashed](0,2){1.4142}{270}{0}
\psdots(1,1.644584)
\uput{7pt}[100](1,1.644584){$R$}
\endpspicture
\end{center}
\caption{}
\label{F:1}
\end{figure}

\begin{proof}

This is one case where a picture is worth one thousand words. We refer to Figure~\ref{F:1} for a visual explanation of why the argument works. The dotted and dashed circle arcs are centred on the vertices and have radius $1/\sqrt2$. The plain circle arcs have the following property: for any point $P$ on them,   the two square vertices $V$ and $V'$ belonging to them make an angle of $\theta$ with $P$, in other terms $|\widehat{V P V'}| =\theta$. Therefore points between two bigger arcs (in one of the two almond-shaped regions) ``look'' at the diagonally opposite vertices marking the intersection of these arcs with an angle between $\theta$ and $\pi$. We then choose the closest vertex to get a distance $\leq 1/\sqrt2$. In case $P$ is at the intersection of the two almond-shaped regions (in the ``blown square''), we may have to choose one region where one of the vertices is at distance $\leq 1/\sqrt2$, but this is always possible, since the dashed and dotted disks cover everything. Finally, if $P$ does not belong to the union of the two almond-shaped regions, then it lies inside one of the smaller plain disks, where its angle between two appropriate consecutive vertices will also be between $\theta$ and $\pi$. Furthermore, by choosing the closest vertex $V_1$ to $P$, we have $V_1P < 1/\sqrt2$.  \qed
\end{proof}

It remains to explain how we can calculate $\theta$, or rather its value on the usual trigonometric functions $\sin\theta$ and $\cos\theta$ (which is what we really need later), since we can show that they are algebraic numbers expressible by radicals, but $\theta/\pi\notin \mathbb{Q}$.

We concentrate on finding the cartesian coordinates of $R=(1/2,1-u/2)$, appearing in Figure~\ref{F:1}, supposing the vertices are the origin, $(1,0), (1,1)$ and $(0,1)$. Our aim is then to find $u\in(0,1)$. A look at Figure~\ref{F:2} shows the disposition of the angles, so that $u= \cot(\theta/2)$ and $2-u = \cot(3\theta/2-\pi) = \cot(3\theta/2)$. The triplication formulas for the cotangent then show that $u$ satisfies the equation
$$
u+\frac{3u-u^3}{1-3u^2} =2 \quad\Longleftrightarrow\quad
2u^3-3u^2-2u+1=0 \enspace.
$$
Solving it yields that the root we are looking for is
$$
u= \frac{{\left( {\sqrt{ 3 } i} + 1 \right) {\left( {{12 \sqrt{ 237 }} i} - 54 \right)}^{\frac{1}{3}} }}{{12 {2}^{\frac{1}{3}} }} + \frac{{{7 {2}^{\frac{1}{3}} } \left( 1 - {\sqrt{ 3 } i} \right)}}{{4 {\left( {{12 \sqrt{ 237 }} i} - 54 \right)}^{\frac{1}{3}} }} + \frac{1}{2} \approx 0.3554157
$$
where the determination of the cube root is the one in the first quadrant.

\begin{figure}
\begin{center}
\psset{unit=2.5}
\pspicture(-1,0)(3,2)
\pspolygon(0,0)(2,0)(2,2)(0,2)
\psdots(1,1.644584)
\uput{7pt}[120](1,1.644584){$R$}
\psline[linestyle=dotted](1,0)(1,2)
\psline[linestyle=dotted](0,1.644584)(2,1.644584)
\psline(1,1.644584)(2,2)
\psline(0,2)(1,1.644584)
\psline(1,1.644584)(2,0)
\radians
\psarc(1,1.644584){0.5}{-1.5708}{-1.02447}
\psarc(1,1.644584){0.2}{0.34149}{1.5708}
\psarc(1,1.644584){0.2}{2.8}{-1.02447}
\uput{0.7}[-1.3](1,1.644584){$\frac{3\theta}2 -\pi$}
\uput{0.25}[0.7](1,1.644584){$\theta/2$}
\uput{0.25}[-2.5](1,1.644584){$\theta$}
\endpspicture
\end{center}
\caption{}
\label{F:2}
\end{figure}

\begin{remark}
One can see that $\theta/\pi\notin\mathbb{Q}$ in the following way. 
$$
\cot(\theta/2)= i\, \frac{e^{i\theta/2} + e^{-i\theta/2}}{e^{i\theta/2}-e^{-i\theta/2}}
$$
and supposing by absurd that $\theta/\pi\in\mathbb{Q}$ we would have that $e^{i\theta}$ is a root of unity. The preceding equality shows that then $\cot(\theta/2)$ belongs to a cyclotomic extension of $\mathbb{Q}$, whose Galois group is abelian. But we have seen that the irreducible polynomial of $\cot(\theta/2)$ is $2x^3-3x^2-2x+1$, with discriminant $316$, not a rational square. Therefore its Galois group is the nonabelian $S_3$, contradiction.
\end{remark}


\begin{remark}
When applying Lemma~\ref{L:geo}, it is essential that we be able to choose from the set of all vertices of the lattice square which ones are the adequate $V_1$ and $V_2$. Since the only excluded quotient $q_j$ is zero, it means that we must be careful to avoid all four squares which have the origin as a vertex. But this follows from the fact that at all steps $j\geq 0$ we always have $|r_j/r_{j+1}| \geq \sqrt 2$.
\end{remark}

Define $\Theta = \arctan2 - \pi/3$ and $A=1/\sin\Theta= 2\sqrt5\,(8+5\sqrt3)/\sqrt{13+4\sqrt3} \approx 16.6902$.
In the following analysis of the EGEA, it will be useful to make the following distinction between indices.

\begin{definition}[Good and bad $j$'s]
A step $j\geq 0$ of the EGEA will be called \emph{bad} if, during the $j-1$-th step, among all four choices of $q_{j}$ as a vertex of the lattice square containing $r_{j-1}/r_j$  (and consequent choice of $r_{j+1}$ and $s_{j+1}$, noting that for the purpose of this definition we do not require that $|r_{j+1}|<|r_j|$), we always have $s_j s_{j+1}r_{j+1}\neq 0$ and 
$$
\left| \arg\left( \frac{ s_{j+1}r_j}{s_j r_{j+1}} \right) \right | < \Theta \enspace.
$$
Otherwise $j$ is called \emph{good}.
\end{definition}

\begin{remark}
Note that $j=0$ and $j=1$ are always good, since $s_1=0$.
\end{remark}

\begin{lemma}[Use of good $j$'s]
\label{L:good}
If $j$ is good then for some choice of $r_{j+1}'$ (and relative $s_{j+1}'$) we have
$$
\left |s_{j+1}'r_j - s_jr_{j+1}' \right| \geq \sin\Theta \max(|s_{j+1}' r_j|, |s_j r_{j+1}'|) 
$$
and, therefore,  if we choose a $r_{j+1}$ with smallest modulus, then
$$
\max(|s_{j}r_{j+1}|, |s_{j+1}r_j|)\leq (A+1) |\nu|
$$
\end{lemma}

\begin{proof}
Notice that the result holds trivially if $s_j s_{j+1}'r_{j+1}'=0$. Otherwise, this is a straightforward application of a general inequality about complex numbers that we can express as follows: let $\zeta \in\mathbb{C}^*$ with $\pi\geq |\arg(\zeta)|\geq\Theta$. We claim that under these conditions,
$|1-\zeta| \geq  \sin\Theta $. Indeed, writing $\zeta = r e^{i\psi}$ with $\Theta \leq \psi\leq \pi$ we have
$$
|1-\zeta|^2 = \left(1-re^{i\psi}\right) \left(1-re^{-i\psi}\right) = 1- 2r \cos\psi + r^2\enspace.
$$
First note that we can suppose $\psi\leq \pi/2$, otherwise clearly $|1-\zeta|\geq 1$. The last expression in $r$, when viewed as a quadratic polynomial has minimum (over $\mathbb{R}$) equal to
$-\Delta/4 = -(4\cos^2\psi -4)/4 = \sin^2\psi \geq \sin^2\Theta$. Therefore $|1-\zeta|\geq \sin\Theta$ thereby proving our claim. The first part of the lemma will follow by applying the claim to $\zeta = s_{j+1}'r_j/s_jr_{j+1}'$ and $\zeta=s_jr_{j+1}'/s_{j+1}'r_j$ successively.

The second part follows from 
$$
|s_{j}r_{j+1}| \leq |s_{j} r_{j+1}'| \leq A \left |s_{j+1}'r_j - s_jr_{j+1}' \right| = A |\nu|
$$
and therefore
$$|s_{j+1}r_j | = |s_{j+1} r_j - s_j r_{j+1} + s_j r_{j+1}| \leq |s_{j+1} r_j - s_j r_{j+1} | + |s_j r_{j+1}|
\leq |\nu| + A |\nu|$$ 
\qed
\end{proof}

\begin{remark}
\label{R:technical}
We have seen in the course of the proof the preceding lemma the following fact: if $\zeta \in\mathbb{C}^*$ with $\pi\geq |\arg(\zeta)|\geq\psi$, then $|1-\zeta| \geq  \sin\psi $. This is equivalent to the following assertion (set $\zeta=1-\xi$), used in the proof of the next lemma: if $|\xi|< \sin\psi$, then $|\arg(1-\xi)| < \psi$.
\end{remark}

The next result is crucial in controlling what happens when things go ``uncontrolled''. Its proof is rather elaborate.

\begin{lemma}[Bad-$j$ behaviour of $s_j$]
\label{L:bad}
If $j$ is bad, then
$$
 |s_{j+1}| \leq 2\sqrt2 \, |s_{j-1}| \quad\text{and} \quad |s_j|  \leq  |s_{j-1}| \quad \enspace.
$$
\end{lemma}

\begin{proof}
We first suppose that the point $P$ of affix $r_{j-1}/r_j$ does not belong to an exceptional lattice square. Let $V_1$ and $V_2$ as in Lemma~\ref{L:geo} of affixes respectively $q_{j}$ and $q_j'$. Upon defining
$r_{j+1}'= r_{j-1} - q_j' r_j$, since $r_{j+1}= r_{j-1} - q_j r_j$, Lemma~\ref{L:geo} states that $\pi\geq
\bigl |\arg\bigl( (q_j - r_{j-1}/r_j)/(q_j' - r_{j-1}/r_j) \bigr) \bigr| =
|\arg(r_{j+1}/r_{j+1}')| \geq \theta$. By definition of  ``bad'' we have, denoting $s_{j+1}' = s_{j-1} - q_j' s_j$,
$$
\left| \arg\left( \frac{ s_{j+1}r_j}{s_j r_{j+1}} \right) \right | < \Theta \quad\text{and}\quad
\left| \arg\left( \frac{ s_{j+1}'r_j}{s_j r_{j+1}'} \right) \right | < \Theta \enspace,
$$
and this yields 
\begin{align*}
\left| \arg\left(\frac{s_{j+1} r_{j+1}'}{s_{j+1}'r_{j+1}}\right)\right| &=
 \left| \arg\left( \frac{ s_{j+1}r_j}{s_j r_{j+1}} \right) +   \arg\left( \frac{s_j r_{j+1}'}{ s_{j+1}'r_j} \right) \right | \\
&\leq \left| \arg\left( \frac{ s_{j+1}r_j}{s_j r_{j+1}} \right) \right| +  \left| \arg\left( \frac{s_j r_{j+1}'}{ s_{j+1}'r_j} \right) \right| < 2\Theta \enspace.
\end{align*}
We deduce 
\begin{align*}
\left| \arg\left(\frac{s_{j+1} r_{j+1}'}{s_{j+1}'r_{j+1}}\right) + \arg \left( \frac{r_{j+1}}{r_{j+1}'}\right)\right| 
&\geq \left| \,  \left| \arg \left( \frac{r_{j+1}}{r_{j+1}'}\right)\right|-  \left |\arg\left(\frac{s_{j+1} r_{j+1}'}{s_{j+1}'r_{j+1}}\right) \right| \,\right| \\
&> \theta - 2\Theta > \frac{2\pi}3 \enspace,
\end{align*}
while on the other hand
$$
\left| \arg\left(\frac{s_{j+1} r_{j+1}'}{s_{j+1}'r_{j+1}}\right)\right| + \left|\arg \left( \frac{r_{j+1}}{r_{j+1}'}\right)\right| 
<  2\Theta+\pi< \theta-\frac{2\pi}3 +\pi < \frac{4\pi}3 
$$
which together imply
\begin{equation}
\label{E:3pi4}
\left| \arg \left( \frac{s_{j+1}}{s_{j+1}'}\right) \right| > \frac{2\pi}3 \enspace.
\end{equation}

Now assume that $|s_{j}| > |s_{j-1}|$. Then $|q_j s_j|>\sqrt2\, |s_{j-1}|$ and $|q_j' s_j| > \sqrt2\, |s_{j-1}|$, since the quotients $q_j, q_j'$ Gaussian integers of modulus different from zero or one. Furthermore, since there is at most one Gaussian integer of modulus one in a lattice square, we have that either  $|q_j s_j|>2 |s_{j-1}|$ or $|q_j' s_j|>2 |s_{j-1}|$. Therefore, by Remark~\ref{R:technical},
$$
\left| \arg \left(\frac{q_js_j- s_{j-1}}{q_j s_j}\right)\right |
= \left| \arg \left(1- \frac{s_{j-1}}{q_j s_j}\right)\right | \leq \frac\pi4
$$
and similarly
$$
\left| \arg \left(\frac{q_j's_j- s_{j-1}}{q_j' s_j}\right)\right |
\leq \frac\pi4 \enspace,
$$
with at least one of them being $\leq \pi/6$.
We then get, using that $|\arg(q_j/q_j')|\leq \pi/4$,
\begin{align*}
\left| \arg \left( \frac{s_{j+1}}{s_{j+1}'}\right) \right| &= 
\left| \arg \left(\frac{q_js_j- s_{j-1}}{q_j s_j}\right) + 
\arg\left( \frac{q_j}{q_j'}\right) + \arg \left(\frac{q_j' s_j}{q_j's_j- s_{j-1}}\right)
\right | \\
&\leq \frac{\pi}4 + \frac\pi4 + \frac\pi6 = \frac{2\pi}3
\end{align*}
contradicting~\eqref{E:3pi4}. 

\noindent\textbf{Exceptional Choices of $V_1$ and $V_2$}: \newline
We now discuss the exceptional cases when $q_{j}$ or $q_j'=\pm1,\pm i$, which need to be handled \textit{ad hoc}. By symmetry (without loss of generality), we place ourselves in the case when $r_{j-1}/r_{j}$ lies in the lattice square of vertices $i, 1+i, 1+2i, 2i$. It then belongs to one of the five regions labeled $1$ to $5$ on Figure~\ref{F:3}. Note that the open grey region is off-limits, since $|r_{j-1}/r_{j}| \geq \sqrt 2$. Each of these regions contains two lattice points, which as before we will denote $V_1$ for the one closest to the point $P$ of affix $r_{j-1}/r_{j}$ and $V_2$ for the other one (in case $P$ lies on the boundary between two or more zones, their distinction is immaterial). Note that $V_1$ is closest to $P$ among \emph{all} four vertices.

\begin{figure}
\begin{center}
\psset{unit=2.5}
\pspicture(-1,0)(3,2)
\pspolygon[fillstyle=vlines](0,0)(0.5,1)(0,2)
\pspolygon[fillstyle=vlines, hatchangle=0](0,2)(1,1.5)(2,2)
\pspolygon[fillstyle=hlines](2,2)(1.5,1)(2,0)
\pspolygon[fillcolor=gray,fillstyle=solid](0,0)(1,0.5)(2,0)
\pspolygon[linestyle=none,fillstyle=solid, fillcolor=red](0,0)(1,0.5)(1,1.5)(2,2)(1.5,1)(0.5,1)
\pspolygon[linestyle=none,fillstyle=solid, fillcolor=yellow](2,0)(1,0.5)(1,1.5)(0,2)(0.5,1)(1.5,1)
\uput{0.05}[-135](0,0){$i$}
\uput{0.05}[-45](2,0){$1+i$}
\uput{0.05}[45](2,2){$1+2i$}
\uput{0.05}[135](0,2){$2i$}
\psdots(0,0)(2,0)(2,2)(0,2)
\uput*[0](0.1,1){1}
\uput*[180](1.9,1){3}
\uput*[-90](1,1.9){2}
\uput{0.5}[135](1,1){5}
\uput{0.5}[-45](1,1){5}
\uput{0.5}[45](1,1){4}
\uput{0.5}[-135](1,1){4}
\psline[linestyle=dotted](1,0)(1,2)
\psline[linestyle=dotted](0,1)(2,1)
\endpspicture
\end{center}
\caption{}
\label{F:3}
\end{figure}

\noindent\textbf{Region} 1 (delimited by a triangle of vertices $i, 1/4+3i/2, 2i$):
In this case, $ |\arg(r_{j+1}/r_{j+1}')| = \widehat{V_1PV_2} \geq 2\arctan 2$. Supposing to fix notations that $q_j=i$ and $q_j'=2i$ we have, assuming that $|s_j|> |s_{j-1}|$ and using Remark~\ref{R:technical}
$$
\left| \arg \left(\frac{q_js_j- s_{j-1}}{q_j s_j}\right)\right |
\leq \frac\pi2 \enspace, \quad \left| \arg \left(\frac{q_j's_j- s_{j-1}}{q_j' s_j}\right)\right |
\leq \frac\pi6 \enspace.
$$
On the other hand, a reasoning similar to the one leading to~\eqref{E:3pi4} with the value $2\arctan 2$ instead of $\theta$ will show that again $ |\arg(s_{j+1}/s_{j+1}')| > 2\pi/3$, which leads to a contradiction since
\begin{align*}
 \left| \arg \left( \frac{s_{j+1}}{s_{j+1}'}\right) \right| &= 
\left| \arg \left(\frac{q_js_j- s_{j-1}}{q_j s_j}\right) + 
\arg\left( \frac{q_j}{q_j'}\right) + \arg \left(\frac{q_j' s_j}{q_j's_j- s_{j-1}}\right)
\right | \\
&\leq \frac{\pi}2 + 0 + \frac\pi6 = \frac{2\pi}3 \enspace.
\end{align*}
The other four cases are treated similarly and we briefly outline them.

\noindent\textbf{Region} 2 (delimited by a triangle of vertices $2i, 1+7i/4, 1+2i$):
Here $ |\arg(r_{j+1}/r_{j+1}')| \geq 2\arctan 2$. Letting $q_j=2i, q_j'=1+2i$ one can show, assuming that $|s_j|> |s_{j-1}|$, that
\begin{align*}
\frac{2\pi}3 < \left| \arg \left( \frac{s_{j+1}}{s_{j+1}'}\right) \right| &= 
\left| \arg \left(\frac{q_js_j- s_{j-1}}{q_j s_j}\right) + 
\arg\left( \frac{q_j}{q_j'}\right) + \arg \left(\frac{q_j' s_j}{q_j's_j- s_{j-1}}\right)
\right | \\
&\leq \frac{\pi}6 +\left( \frac\pi2 - \arctan2\right) +  \left( \frac\pi2 - \arctan2\right) <\frac\pi2< \frac{2\pi}3 \enspace,
\end{align*}
contradiction.

\noindent\textbf{Region} 3 (delimited by a triangle of vertices $1+2i, 3/4+i/2, 1+i$):
Here $ |\arg(r_{j+1}/r_{j+1}')| \geq 2\arctan 2$. Letting $q_j=1+i, q_j'=1+2i$ one can show, assuming that $|s_j|> |s_{j-1}|$, that
\begin{align*}
\frac{2\pi}3 < \left| \arg \left( \frac{s_{j+1}}{s_{j+1}'}\right) \right| &= 
\left| \arg \left(\frac{q_js_j- s_{j-1}}{q_j s_j}\right) + 
\arg\left( \frac{q_j}{q_j'}\right) + \arg \left(\frac{q_j' s_j}{q_j's_j- s_{j-1}}\right)
\right | \\
&\leq  \frac{\pi}4 + \arctan(1/3) +\left( \frac\pi2 - \arctan2\right) 
=\frac\pi2< \frac{2\pi}3 \enspace,
\end{align*}
contradiction.

\noindent\textbf{Region} 4 (the red zone):
Here $ |\arg(r_{j+1}/r_{j+1}')| \geq \pi-\arctan2 + \arctan(2/3)$. Letting $q_j=i, q_j'=1+2i$ one can show, assuming that $|s_j|> |s_{j-1}|$, that
\begin{align*}
\frac{5\pi}3-3\arctan(2) + \arctan(2/3) &< \left| \arg \left( \frac{s_{j+1}}{s_{j+1}'}\right) \right| \\
&= 
\left| \arg \left(\frac{q_js_j- s_{j-1}}{q_j s_j}\right) + 
\arg\left( \frac{q_j}{q_j'}\right) + \arg \left(\frac{q_j' s_j}{q_j's_j- s_{j-1}}\right)
\right | \\
&\leq \frac{\pi}2 + \left( \frac\pi2 - \arctan2\right) +\left( \frac\pi2 - \arctan2\right) \\
&<\frac{5\pi}3-3\arctan(2) + \arctan(2/3)\enspace,
\end{align*}
contradiction.

\noindent\textbf{Region} 5 (the yellow zone):
Here $ |\arg(r_{j+1}/r_{j+1}')| \geq \pi - \arctan(4/7)$. Letting $q_j=1+i$, $q_j'= 2i$ one can show, assuming that $|s_j|> |s_{j-1}|$, that
\begin{align*}
\frac{5\pi}3-3\arctan(2) + \arctan(2/3) &< \left| \arg \left( \frac{s_{j+1}}{s_{j+1}'}\right) \right| \\
&= 
\left| \arg \left(\frac{q_js_j- s_{j-1}}{q_j s_j}\right) + 
\arg\left( \frac{q_j}{q_j'}\right) + \arg \left(\frac{q_j' s_j}{q_j's_j- s_{j-1}}\right)
\right | \\
&\leq \frac{\pi}4 + \frac\pi4 +\frac\pi6= \frac{2\pi}3 \\
&< \frac{5\pi}3-3\arctan(2) + \arctan(2/3) \enspace,
\end{align*}
contradiction.

We have thus proved that in any case $|s_j|\leq |s_{j-1}|$. To show the first part of the lemma, we proceed similarly, although there is a slight difference. We assume at first that $|s_{j+1}|>2\, |s_{j-1}|$ \emph{and} $|s_{j+1}' |> \sqrt2\, |s_{j-1}|$.
Then, by Remark~\ref{R:technical},
$$
\left| \arg \left( \frac{s_{j+1} - s_{j-1}}{s_{j+1}} \right) \right | =
\left| \arg \left( 1 -\frac{s_{j-1}}{s_{j+1}} \right) \right | \leq  \frac\pi6 
$$
and
$$
\left| \arg \left( \frac{s_{j+1}' - s_{j-1}}{s_{j+1}'} \right) \right | 
\leq  \frac\pi4 \enspace.
$$ 
Proceeding as previously,
\begin{align*}
\left| \arg \left( \frac{s_{j+1}}{s_{j+1}'}\right) \right| &= 
\left| \arg \left( \frac{s_{j+1}}{s_{j+1} - s_{j-1}} \right) 
+ \arg\left( \frac{q_j}{q_j'}\right) +
 \arg \left( \frac{s_{j+1}' - s_{j-1}}{s_{j+1}'} \right) 
\right | \\
&\leq \frac\pi4 + \frac\pi4 + \frac\pi6 = \frac{2\pi}3
\end{align*}
again contradicting~\eqref{E:3pi4}, which also holds in the exceptional cases, as we have just seen.
Therefore $|s_{j+1}| \leq 2\, |s_{j-1}|$ \emph{or} $|s_{j+1}' | \leq \sqrt2\, |s_{j-1}|$ (or both). In the first case, we are done. Otherwise,
since $s_{j+1}+ q_j s_j = s_{j+1}' + q_j' s_j=s_{j-1}$, we derive
$$
|s_{j+1}|\leq |s_{j+1}'| + |q_j - q_j'| |s_j| \leq \sqrt2\, |s_{j-1}| + \sqrt 2 |s_{j-1}| = 2\sqrt2\, |s_{j-1}|\enspace,
$$
by the already proved second part of the lemma and the fact that $q_j, q_j'$ correspond to two vertices of the same lattice square, so that $|q_j -q_j'| \leq \sqrt 2$.
\qed
\end{proof}

\begin{lemma}[Lower bound on generic vectors of $\ker F$]
\label{L:genlow}
For any nonzero $(z_1,z_2)\in\ker F$ we have
$$
\max(|z_1|, |z_2|) \geq \frac{\sqrt{|\nu|}}{\sqrt{1+|r|+s}} \enspace.
$$
In particular, for any $j\geq 0$ we have
$$
\max(|r_j|, |s_j|) \geq \frac{\sqrt{|\nu|}}{\sqrt{1+|r|+s}} \enspace.
$$
\end{lemma}

\begin{proof}
This proof uses an argument already appearing in the proof of the original GLV algorithm, see~\cite{SCQ02}, as well as Lemma~\ref{L:genlownumfield}. If $(0,0)\neq(z_1, z_2)\in \ker F$ then $z_1+\lambda z_2 \equiv 0 \pmod\nu$. If $\lambda'$ is the other root of $X^2+rX+s \pmod n$, we get that 
$$
z_1^2 - r z_1z_2 +s z_2^2 \equiv (z_1+\lambda z_2) (z_1 +\lambda' z_2) \equiv 0 \pmod\nu \enspace.
$$
Since $X^2+rX+s$ is irreducible in $\mathbb{Q}(i)$ because the two quadratic fields are linearly disjoint, we therefore have $|z_1^2 - r z_1z_2 +s z_2^2|\geq |\nu|$. On the other hand if 
$$
\max(|z_1|, |z_2|) < \frac{\sqrt{|\nu|}}{\sqrt{1+|r|+s}} \enspace,
$$
then 
$$
|z_1^2 - r z_1z_2 +s z_2^2| \leq |z_1|^2 + |r| |z_1| |z_2| + s |z_2|^2 
< |\nu|\enspace,
$$
a contradiction. To show the second part, it suffices to note that since $r_0 s_j + \nu t_j = r_j$ (where, as mentioned previously, $r_0=\lambda$ or $\lambda+n$), we have that 
$$
0\equiv \nu t_j= r_j - r_0 s_j \equiv r_j - \lambda s_j \pmod\nu
$$
so that $(r_j, -s_j)\in \ker F$ for every $j\geq0$.
\qed
\end{proof}

\begin{proof}[of Theorem~\ref{T:3}]
It remains here to show the improved bound, which brings us to finding four $\mathbb{Q}$-linearly independent vectors of $\ker F$ of rectangle norm bounded by $Cn^{1/4}$.
Define $m\geq 1$ as the index such that 
\begin{equation}
\label{E:r_m+1}
|r_m| \geq \frac{\sqrt{|\nu|}}{\sqrt{1+|r|+s}} \quad \text{and} \quad  |r_{m+1}| < \frac{\sqrt{|\nu|}}{\sqrt{1+|r|+s}}\enspace.
\end{equation}
Let us consider an index $j\leq m$. If it's good, then by Lemma~\ref{L:good} we have $|s_{j+1} r_{j}| \leq (A+1) |\nu|$ and therefore, since $(|r_j|)$ is a decreasing sequence, 
\begin{equation}
\label{E:s_j}
|s_{j+1}| \leq 2\sqrt2(A+1)\sqrt{1+|r|+s} \sqrt{|\nu|}\enspace.
\end{equation}

On the other hand if it's bad, then let $l < j$ be the largest good index less than $j$. By Lemma~\ref{L:bad} and Lemma~\ref{L:good} we have 
\begin{align*}
\frac{|s_{j+1}|}{2\sqrt2}\leq |s_{j-1}| \leq |s_{j-2}| \leq \dots \leq |s_l | &\leq (A+1) \frac{ |\nu|}{|r_{l+1}|}\\
&\leq (A+1) \sqrt{1+|r|+s} \sqrt{|\nu|} \enspace,
\end{align*}
therefore in any case~\eqref{E:s_j} holds. Applying this to $j=m-1$ and $j=m$ we find that
\begin{equation}
\label{E:s}
\max(|s_m|, |s_{m+1}|) \leq 2\sqrt2(A+1) \sqrt{1+|r|+s} \sqrt{|\nu|}\enspace.
\end{equation}
Moreover, using 
\begin{equation}
\label{E:mix_m}
s_{m+1} r_m - s_m r_{m+1} = (-1)^{m+1} \nu
\end{equation}
and from~\eqref{E:r_m+1}, \eqref{E:s} we deduce
$$
|s_{m+1}r_m| \leq {|\nu| + |s_m r_{m+1}| }  \leq |\nu| + 2\sqrt2(A+1) |\nu| \enspace.
$$
In addition, by Lemma~\ref{L:genlow} we must have 
$$
|s_{m+1}| \geq \frac{\sqrt{|\nu|}}{\sqrt{1+|r|+s}}
$$
which therefore implies that 
$$
|r_m| \leq \bigl(2\sqrt2\left(A + 1\right)+1\bigr) \sqrt{1+|r|+s} \sqrt{|\nu|}\enspace.
$$
This last equation, together with~\eqref{E:s} and~\eqref{E:r_m+1}, show that the two vectors $v_1=(r_m, -s_m), v_2=(r_{m+1}, -s_{m+1}) \in \ker F\subset \Z[i]^2$ have rectangle norms bounded by $C \sqrt{|\nu|} = C n^{1/4}$, for $C=  \bigl(2\sqrt2\left(A + 1\right)+1\bigr) \sqrt{1+|r|+s} < 51.5\, \sqrt{1+|r|+s}$ (that these two vectors belong to $\ker F$ was shown in the proof of Lemma~\ref{L:genlow}). 

We can find two more vectors by noticing that~\eqref{E:mix_m} implies that $v_1$ and $v_2$ are $\mathbb{Q}(i)$-linearly independent. Therefore, the vectors $v_1, v_2, v_3=iv_1, v_4=iv_2$ are $\mathbb{Q}$-linearly independent. They all belong to $\ker F$ and have rectangle norms bounded by $C n^{1/4}$. In view of the fact that the Euclidean norm upper-bounds the rectangle norm, the corresponding vectors in $\Z^4$ also have rectangle norms bounded by $C n^{1/4}$, thus concluding the proof of the theorem, since these are exactly the four vectors returned by Algorithm~\ref{A:3}.
\qed
\end{proof}


\begin{remark}
Let us note that since we are in dimension less than 5, {\fontencoding{T5}\selectfont Nguy\~\ecircumflex n} and Stehl\'e~\cite{NS04} have also produced and algorithm which finds vectors of successive minima in $\ker F$ with a running time $O(\log^2 n)$. However it doesn't seem to give an explicit bound on their length applicable to our case.
\end{remark}

\begin{remark}
One may doubt about the pertinence of securing a faster lattice reduction algorithm for the GLV lattice $\ker F$. Indeed, at the present moment, the GLV method has been applied by choosing a fixed curve in the parameters and performing the lattice reduction offline. However, it is quite possible that in the future some new cryptosystem will require an online curve agreement, by counting points over a suitable field and successively performing the lattice reduction. For this, and human tendency of pushing away limitations, we consider that our previous argument, in addition to its formal elegance, may eventually find a useful application.
\end{remark}

\section{Performance Estimates}

In this section, we assess performance of the four-dimensional GLV method in comparison with the traditional and two-dimensional cases. For our analysis, let us consider the curve $E: y^2 = x^3+b$ over a quadratic extension field of large prime characteristic, exploiting a pseudo-Mersenne prime $p$ such that $-1$ is a quadratic non-residue mod $p$, for efficiency purposes.  
Let us define the following notation: i) $M, S, A$ and $I$ represent field multiplication, squaring, addition and inversion over $\mathbb{F}_p$, respectively, and ii) $m, s, a$ and $i$ represent the same operations over $\mathbb{F}_{p^2}$. On the curve above, the expected costs of scalar multiplication at the 128-bit security level in terms of $\mathbb{F}_{p^2}$ operations on one processor core are given by

\begin{itemize}
\item One core non-GLV: $256\mathrm{DBL} + 42.5\mathrm{mADD} + C_{\mathrm{\emph{Precomp}}} +  C_{\mathrm{\emph{Affine}}} = 1108m + 1152s + 2090a + (1i + 64m + 19s + 56a) + (1i + 3m + 1s) = 2i + 1175m + 1172s + 2146a$,
\item One core 2-GLV: $128\mathrm{DBL} + 43.5\mathrm{mADD} + C_{\mathrm{\emph{Precomp}}} +  C_{\mathrm{\emph{Affine}}} = 732m + 643s + 1201a + (1i + 78m + 19s + 63a) + (1i + 3m + 1s) = 2i + 813m + 663s + 1264a$,
\item One core 4-GLV: $64\mathrm{DBL} + 45.5\mathrm{mADD} + C_{\mathrm{\emph{Precomp}}} +  C_{\mathrm{\emph{Affine}}} = 556m + 393s + 767a + (1i + 106m + 19s + 77a) + (1i + 3m + 1s) = 2i + 665m + 413s + 844a$,
\end{itemize}

\noindent corresponding to the fully sequential executions without using GLV, using 2-GLV and using 4-GLV, respectively. Note that DBL and mADD represent point doubling and mixed addition, and that their costs are $3m+4s+7a$ and $8m+3s+7a$ when using Jacobian coordinates. $C_{\mathrm{\emph{Precomp}}}$ and $C_{\mathrm{\emph{Affine}}}$ represent the cost of precomputation and final conversion to affine coordinates, respectively. The costs above assume the use of interleaving (INT) \cite{GLV01} with width-$w$ non-adjacent form ($w$NAF) using 7 precomputed points ($w=5$) and the use of the LM scheme for precomputing those points \cite{LM08}.
 
Similarly, the expected cost of 4-GLV on four cores is

\begin{itemize}
\item Four core 4-GLV: $64\mathrm{DBL} + 12.5\mathrm{mADD} + C_{\mathrm{\emph{Precomp}}} +  C_{\mathrm{\emph{Affine}}} = 292m + 294s + 536a + (1i + 69m + 19s + 58a) + (1i + 3m + 1s) = 2i + 364m + 314s + 594a$.
\end{itemize}

Thus, it can be seen a steady cost reduction when switching from a non-GLV based implementation to 2- and 4-GLV. A similar improvement is observed when increasing the number of cores in the case of 4-GLV. Optimal performance is ultimately achieved with this method running on four cores. For instance, following implementation results on an AMD Phenom II X4 940, we have that $1i \approx 50m$, $1s \approx 0.70m$ and $1a \approx 0.2m$ when using a 127-bit prime. In this case, non-GLV, 2-GLV and 4-GLV cost $2525m$, $1630m$ and $1223m$ on one core, respectively. On the other hand, four core 4-GLV costs $803m$. Hence, a remarkable 3x speedup is expected when using 4-GLV on four cores in comparison with a traditional execution on one core.   

Note that, in comparison to $\mathbb{F}_p^2$ arithmetic using a 128-bit prime (such as the one used in \cite{HLX11}), $\mathbb{F}_p^2$ multiplication is expected to be faster since internal field additions can be carried out without carry checks and lazy reduction applies efficiently (see also \cite{Lon11b}); however, other operations get slightly more expensive because reduction involves a few extra shift and rotate operations. 

Let us compare performance against a similar curve using the original GLV method over $\mathbb{F}_p$. In this case the expected costs when using one and two cores are ($1I \approx 200M$, $1S \approx 0.85M$ and $1A \approx 0.2M$ on the targeted platform)

\begin{itemize}
\item One core standard 2-GLV: $128\mathrm{DBL} + 37.2\mathrm{ADD} + 5.3\mathrm{mADD} + C_{\mathrm{\emph{Precomp}}} +  C_{\mathrm{\emph{Affine}}} = 836M + 640S + 1194A + (51M + 26S + 56A) + (1I + 3M + 1S) = 1I + 890M + 667S + 1250A = 1907M$,
\item Two core standard 2-GLV: $128\mathrm{DBL} + 19\mathrm{ADD} + 3.5\mathrm{mADD} + C_{\mathrm{\emph{Precomp}}} +  C_{\mathrm{\emph{Affine}}} = 621M + 580S + 1054A + (44M + 26S + 56A) + (1I + 3M + 1S) = 1I+ 668M + 607S + 1110A = 1606M$,
\end{itemize}

\noindent where DBL, mADD and ADD cost $3M+4S+7A$, $8M+3S+7A$ and $11M+3S+7A$ for the case of Jacobian coordinates. We assume the use of $w$NAF with window width $w=5$ and the LM precomputation scheme without inversions \cite[Ch. 3]{Lon11}. Since in our case we have observed that in practice $1M \approx 0.75m$, the scaled costs of one core and two core standard 2-GLV are equivalent to $1430m$ and $1205m$, respectively. This means that on one core the 4-GLV method is expected to compute scalar multiplication in about 0.86 the time of the standard 2-GLV. Similarly, an optimal execution of 4-GLV on four cores is expected to run in about 0.67 the time of the optimal execution of the standard 2-GLV on two cores.      
 
To confirm our findings we implemented the proposed method using the quadratic twist of $E_1$ over $\mathbb{F}_{p_1^2}$ given by $E'_1/\mathbb{F}_{p_1^2}: y^2 = x^3 + u.9$, where $E_1/\mathbb{F}_{p_1}: y^2 = x^3 + 9$, $p_1 = 2^{127} - 58309$, $u$ is a non-square in $\mathbb{F}_{p_1^2}$ and $\#E'_1(\mathbb{F}_{p_1^2})$ is a 254-bit prime. Since $p_1 = 3\pmod 8$, we represent $\mathbb{F}_{p_1^2}$ as $\mathbb{F}_{p_{1}}[i]$, where $i = \sqrt{-1}$. Let $u = 1+i$. The two endomorphisms are given by $\Phi(x,y) = (\xi x, y) = \lambda P$ and $\Psi(x,y) = (u^{(1-p)/3} x, u^{(1-p)/2} y) = \mu P$, where $\xi^3 = 1\mod p_1$. Following Section 5 it can be verified that $\Phi^2 + \Phi + 1 = 0$ and $\Psi^2 + 1 = 0$. Note that a similar curve was also used in \cite{HLX11} but using a different 4-GLV construction. For the case of the standard 2-GLV, we use the curve $E_2/\mathbb{F}_{p_2}: y^2 = x^3 + 2$, where $p_2 = 2^{256}-11733$ and $\#E'_2(\mathbb{F}_{p_2})$ is a 256-bit prime. The endomorphism is given by $\Phi(x,y)$.      

For our experiments we used a 3.0GHz AMD Phenom II X4 940 processor with four cores. The expected timings in terms of clock cycles are displayed in Table 1. As can be seen, closely following our analysis and considering that the use of multiple cores inserts certain penalty, 4-GLV on four cores injects a speed up close to 3x in comparison with a fully sequential version on one core, and supports a computation that runs in 0.67 the time of the standard 2-GLV on two cores.   

\begin{table}[!t]
\begin{center}
\caption{Point multiplication timings (in clock cycles), 64-bit processor}
\begin{tabular}{lcc}
\hline\noalign{\smallskip}
\hspace{20 mm} Method  \hspace{25 mm} &     \# of cores  \hspace{2 mm}  &  AMD Phenom II       \\
\noalign{\smallskip}\hline\noalign{\smallskip}

$E'_1(\mathbb{F}_{p_1^2})$ 127-bit $p$, 4GLV+INT, 7pts.   & \hspace{1 mm}4  & \hspace{1 mm}91,000   \\

$E'_1(\mathbb{F}_{p_1^2})$ 127-bit $p$, 4GLV+INT, 7pts.   & \hspace{1 mm}1  & \hspace{1 mm}124,000   \\

$E'_1(\mathbb{F}_{p_1^2})$ 127-bit $p$, $w$NAF, 7pts.   & \hspace{1 mm}1  & \hspace{1 mm}248,000   \\

$E_2(\mathbb{F}_{p_2})$ 256-bit $p$, 2GLV+INT, 7pts. & \hspace{1 mm}2  & \hspace{1 mm}136,000  \\

\noalign{\smallskip}\hline

\end{tabular}
\end{center}
\end{table}

Although these experimental results correspond to a $j=0$ curve, we can confidently express that the relative improvement from 2-GLV to 4-GLV on any single GLV curve will be by the same order. This follows theoretically from the minimality of Theorem~\ref{T:min} and the form of Theorem~\ref{T:3}, which together imply that $\frac{\text{bitlength of 4-GLV coefficients}}{\text{bitlength of 2-GLV coefficients}} \approx 1/2$ (no $r, s$ involved). 





\section{Conclusion}
We have produced new families of GLV curves, and written all such curves (up to isomorphism) with nontrivial endomorphisms of degree $\leq 3$. We have shown how to generalize the Gallant-Lambert-Vanstone scalar multiplication method by combining it with the Galbraith-Lin-Scott ideas, to perform a proven almost fourfold speedup on GLV curves over $\mathbb{F}_{p^2}$. We have provided a first explicit bound on such a decomposition using the LLL algorithm. We have then refined this bound using a faster new reduction algorithm, which consists in basically two applications of the extended Euclidean algorithm, one in $\Z$ and the other in $\Z[i]$. This allows us to get a relative improvement (on the same curve) from 2-GLV to 4-GLV \emph{independent} of the curve.

\vspace{\baselineskip}

\noindent\textbf{Acknowledgements}:
We would like to thank Mike Scott for advice on a first version of this work. Also, we would like to thank Diego F. Aranha for his advice on multi-core programming and Joppe Bos for his help on looking for efficient chains for implementing modular inversion.

\end{document}